\documentclass[12pt]{amsart}
\usepackage{latexsym}
\usepackage{amsmath,amssymb,mathrsfs,pictex} 
\usepackage{amsmath,amsthm,amsfonts,amscd,eucal}
\usepackage{graphicx}
\usepackage{epsfig}
 \usepackage[pdftex,bookmarks,colorlinks,breaklinks]{hyperref}  
\usepackage{epstopdf} 
\usepackage{cancel}

\numberwithin{equation}{section}

\def\cb{{\mathcal B}}

\def\ch{{\mathcal H}}

\def\bc{{\mathbb C}}

\def\br{{\mathbb R}}

\def\gf{{\mathfrak F}}

\def\b{\beta}
        \def\G{\Gamma}
\def\d{\delta}        \def\D{\Delta}
\def\eps{\varepsilon} 
     
\def\z{\zeta}

\def\l{\lambda}       
\def\m{\mu}

\def\r{\rho}
\def\s{\sigma}

\def\om{\omega}        \def\Om{\Omega}

\def\dss{\displaystyle}

\newtheorem{Thm}{Theorem}[section]

\newtheorem{Prop}[Thm]{Proposition}

\theoremstyle{definition}

\theoremstyle{remark}

\def\di{\mathop{\mathop{\rm d}}\!}

\newcommand{\e}[1]{e^{#1}}

\newcommand{\tr}{\textrm{Tr}}

\newcommand{\ty}[1]{\mathop{\rm {#1}}}
\newcommand{\rmd}{\textrm{d}}

\newcommand{\nn}{\nonumber}

\begin{document}
%
%
%
\title[thermodynamics of quons]
{}
\begin{center}
{\textsc{\textbf{ON THE THERMODYNAMICS OF THE \emph{q}-PARTICLES}}}
\end{center}
\author{Fabio Ciolli}
\address{Fabio Ciolli\\
Dipartimento di Matematica \\
Universit\`{a} di Roma Tor Vergata\\
Via della Ricerca Scientifica 1, Roma 00133, Italy} \email{{\tt
ciolli@mat.uniroma2.it}}
\author{Francesco Fidaleo}
\address{Francesco Fidaleo\\
Institute of Mathematics\\
W\"{u}rzburg University\\
Emil-Fischer-Str$\b$e 40, 97074 W\"{u}rzburg, Germany} \email{{\tt
fidaleo@mat.uniroma2.it}}


\date{}

\begin{abstract}
Since the grand partition function $Z_q$ for the so-called $q$-particles (i.e.,\ quons), $q\in(-1,1)$, cannot be computed by using the standard 2nd quantisation technique involving the full Fock space construction for $q=0$, and its $q$-deformations for the remaining cases, we determine such grand partition functions in order to obtain the natural generalisation of the Plank distribution to $q\in [-1,1]$. We also note the (non) surprising fact that the right grand partition function concerning the Boltzmann case (i.e.,\ $q=0$) can be easily obtained by using the full Fock space 2nd quantisation, by considering the appropriate correction by the Gibbs factor $1/n!$ in the $n$ term of the power series expansion with respect to the fugacity $z$. As an application, we briefly discuss the equations of the state for a gas of free quons or the condensation phenomenon into the ground state, also occurring for the Bose-like quons $q\in(0,1)$.
\end{abstract}

\maketitle

\section{Introduction}

Recently, the investigation of exotic models has been enormously increased with the hope to find some progress in long-standing unsolved problems in the physics of complex models and, in parallel to some other disciplines, relevant for concrete applications like the theory of information.

Concerning the unsolved problems in physics, we certainly mention that, in order to provide a satisfactory mathematical description of the quantum electrodynamics, the latter has predictions obtained via the renormalisation technique, which are in surprisingly perfect accordance with the experiments. For such purposes, the reader is referred to the classical literature 
of the mathematical rigorous approach \cite{Ha,Sc}, and to \cite{CRV,BCRV,BCRV19} for a specific description of quantum electrodynamics.

 The models that aim to study and unify the strong interaction (i.e.,\ quantum chromodynamics) are placed on the same line as the electroweak ones. All such models are called standard and have the same strengths and weaknesses, that is these are in good accordance with the experiments but are not satisfactory from the mathematical point of view. The long-standing problem to unify these three fundamental forces present in nature with the remaining one, that is the gravitation, 
was recently attached with the use of the so-called noncommutative geometry, for example \cite{CL}, is very far from being solved even in a partial~form.

Another direction, indeed connected with applications to the previous questions of quantum field theory, and also involving applications which are certainly relevant for the applied sciences, such as quantum information theory and quantum computing, is the detailed study of the von Neumann entropy and the Araki relative entropy, see for example~\cite{A,CLR,CLRR} for the foundations and recent results. 

It is also worth mentioning various other entropies (i.e.,\ Tsallis entropy \cite{T, T1}),
introduced with the perspective of solving some open problem and to be fruitfully applied to information theory.

Among the  studied models, there are certainly those associated to the so called $q$-particles, or quons, $q\in(-1,0)\bigcup(0,1)$, with the perspective of the extension to the so-called anyons (see e.g.,\ \cite{FMa}, see also \cite{BLW} for a mathematical proposal to manage anyons) corresponding to the case when the parameter $q$ assumes values in some root of the unity, and plektons. 

{
To be more precise, the $q$-deformed particles seem to be related to quantum groups and
quantum algebras, which drew much attention decades ago. Such particles emerge naturally from exactly solvable models in statistical mechanics which acquire the Yang--Baxter equation. 
We also point out that the irreducible representations of $q$-deformed
particles are substantial extensions of the quantum algebra in connections to the braid group statistics.
For such interesting applications of these quons, the reader is referred to the monograph  \cite{BL}.
}

Such exotic $q$-particles are naturally associated to the following commutation relations
{
\begin{equation}
\label{qcra}
{\bf a}_q(f){\bf a}_q^\dagger(g)-q{\bf a}_q^\dagger(g){\bf a}_q(f)=\langle g,f\rangle_\ch I_\ch\,, \quad f,g\in\ch\,,
\end{equation}
$\ch$ being the one-particle space equipped with 
the inner product 
$\langle\, \cdot\, ,\, \cdot\, \rangle_\ch$ which is linear in the first argument, enjoyed by the creators and annihilators acting on the corresponding Fock spaces.

If $\ch=\ell^2(I)$ equipped with the canonical basis given by $e_i(j)=\d_{i,j}$, $i,j\in I$, the relations \eqref{qcra} assume the well-known form (cf.\  \cite{BS}, Section 3)
\begin{equation} 
\label{qcra10}
{\bf a}_{q,i}{\bf a}_{q,j}^\dagger-q{\bf a}_{q,j}^\dagger{\bf a}_{q,i}=\d_{i,j} I_{\ell^2(I)}\,, \quad i,j\in I\,.
\end{equation}   

We also note that there are many other deformed commutation relations similar to~\eqref{qcra10} describing $q$-particles (see e.g.,\ \cite{BL}), limiting our analysis only to the mostly studied commutation relations \eqref{qcra}.
}

The quons can certainly be viewed as an interpolation between particles obeying the Fermi statistics (i.e.,\ $q=-1$) and those obeying the Bose statistics (i.e.,\ $q=1$), passing for the value $q=0$ describing the classical particles, and so obeying the Boltzmann statistics.

{
It appears clear that, for the applications to statistical mechanics, the case $q=0$ would correspond to the classical framework as can easily be deduced from \eqref{qcra1} below, where 
$n_q(\eps)\big|_{q=0}=ze^{-\b \eps}$ corresponding to the occupation number of classical particles of the energy level $\eps$ at inverse temperature $\b$ and activity $z$, the commutation relations~\eqref{qcra10} for $q=0$, ${\bf a}_{i}{\bf a}_{j}^\dagger=\d_{i,j}I$, 
have a crucial meaning in the so-called \emph{free probability}, see for example~\cite{S}. 
}

Therefore, in view of the potential applications outlined before, an intensive investigation of these particles was carried out, a consistent part of which regarded the thermodynamics enjoyed by such particles. We mention just a sample \cite{Gr,W,AK,M,IK} of such papers, and refer the reader to the citations therein for further details.

On the other hand, quickly explained above,  the mathematical investigation of the structure and properties of algebras of operators associated to such $q$-models, mainly for the Boltzmann case $q=0$ (called ``free'' in the operator algebra setting, see e.g.,\ \cite{VDN}), was carried out in an intensive way. As a sample of such papers, we also mention \cite{S, BKS} for the applications to quantum probability. The main object of such a mathematical investigation is the full Fock space $\gf(\ch)\equiv \gf_0(\ch)$ for the free-Boltzmann case $q=0$, and the deformed versions of that, the $q$-deformed Fock spaces 
$\gf_q(\ch)$. 

By coming back to the thermodynamics of very huge systems made of particles of the order of the Avogadro number $N_{\rm A}\sim10^{23}$, the main ingredient is the computation of the grand partition function by using the so-called 2nd quantisation method and the relative Fock spaces, see for example \cite{LL, Hu}. 

If, on one hand, this is perfectly suited for the Bose and Fermi situation by using symmetric and totally anti-symmetric Fock spaces, respectively,
on the other hand,
{the standard ingredient to use the full Fock space $\gf_0(\ch)$ for the Boltzmann case $q=0$ and the deformed versions $\gf_q(\ch)$ of that
for  $q\in(-1,0)\bigcup(0,1)$ fails as it is explained in \cite{W}. 

In fact, for $q=0$ the computation of the grand partition function using the 2nd quantized Hamiltonian $K$ in 
\eqref{gfip0} does not take into account the Gibbs paradox, (cf.\ \cite{Hu}), and thus produces the wrong result,  
the right one being \eqref{gfipp}, which is obtained from \eqref{gfip} by correcting with the factor $1/n!$, which takes into account the fact that particles must be considered indistinguishable, see e.g.,\ \cite{FV} (p. 680).

Equation \eqref{gfip} produces another wrong consequence, that is the grand partition function would be defined only for the values of the activity 
$z=e^{\b\mu}<\z^{-1}$ ($\m$ and $\z$ being the chemical potential and partition function, see {\eqref{gfip}}), 
whereas it is well known that, for classical particles, the activity can assume all values $0<z<+\infty$.    

If instead $q\in (-1,0)\bigcup(0,1)$, any $\gf_q(\ch)$ is the deformed version of $\gf_0(\ch)$, and thus the use of such Fock spaces to compute the grand partition functions produces the paradoxical result (cf.\ \cite{W}) that all of those functions coincide with the grand partition function for $q=0$. 
This would mean that the thermodynamics of the $q$-particles, \mbox{$-1<q<1$}, provided that such exotic particles really exist in nature, does not depend on $q$.    

Summarizing, the use of the standard way to compute the gran partition function through the corresponding Fock spaces  $\gf_q(\ch)$, $q\in (-1,1)$, is totally unusable.      
  
However, while the computation of the grand partition function for the classical case ($q=0$) can be easily achieved by taking into account the Gibbs correction $1/n!$, 
$$
\sum_{n=1}^{+\infty}\frac{\big(\tr\,e^{-\b H}\big)^n}{n!}z^n,\,\,\text{instead of}\,\,\sum_{n=1}^{+\infty}\big(\tr\,e^{-\b H}\big)^nz^n\,,
$$
the remaining case, $q\in (-1,0)\bigcup(0,1)$ cannot be overcome with similar methods because    
it is completely unknown what should be the, necessarily deformed, statistics
to which the quons obey.
}

Therefore, the method arising from the 2nd quantisation and using the deformed Fock space introduced in \cite{BKS} is doomed to fail. 

Instead, 
by using a totally different analysis involving essentially only the commutation relations \eqref{qcra}, it was possible to compute, first in \cite{AF} and then in \cite{FV}, the average population $n(\eps)$ of any energy level $\eps$ of the model describing such $q$-particles, obtaining
\begin{equation}
\label{qcra1}
n(\eps)=\frac{1}{z^{-1}e^{\b\eps}-q}\,,\quad -1\leq q\leq1
\end{equation}
{where, for the sake of simplicity, we are supposing that the degeneracy $g(\eps)$ of the energy of any level $\eps$ is 1.
}

It is noteworthy that, for $q=1$ (and $z=1$ as it regards the quantum harmonic oscillator), $\eps=\hbar\om$ and $\b=1/k_B T$, the above formula reduces to the celebrated formula $n=\frac1{e^{\frac{\hbar\om}{k_B T}}-1}$, solving the long standing paradox concerning the 
black-body radiation, for which M. Planck was awarded by the Nobel prize.

{
The formula \eqref{qcra1}, as well as \eqref{occq} below, is standardly recovered for the Bose/Fermi case from $\tr(e^{-\b K_{\pm 1}})$, where $K_{\pm 1}=\di\G_{\pm 1}(H)- \m N$ and $N$ the number operator, by using the well-known formula
$$
 -\frac1{\b}\frac{\partial\,\,}{\partial\eps}\ln\tr\big(e^{-\b (\di\G_{\pm 1} (H)-\mu N)}\big)=n_{\pm 1}(\eps)\,.
$$
Concerning the remaining cases $q\in (-1,1)$, the previous recipe cannot be applied for the consideration just listed above. 

However, \eqref{qcra1} and \eqref{occq} are rigorously recovered, first in \cite{AF} and then in \cite{FV}. Indeed, in Section 5 of \cite{AF}, Equation \eqref{qcra1} is demonstrated by using the continuous analogous of the commutation relations \eqref{qcra10} and imposing the Kubo--Martin--Schwinger (KMS for short) condition. 
In \cite{FV} instead, \eqref{occq} is proved by maximising the $q$-entropy functional, suitably introduced there, by using the Lagrange multipliers method.     
}

Therefore, in the present paper, we start from \eqref{qcra1} to compute the grand partition ${\ty{Z}}_q$ for all values $q\in[-1,1]$, by imposing that
$-\frac1{\b}\frac{\partial\ln{\ty{Z}}_q}{\partial\eps}=n(\eps)$, as prescribed by the statistical mechanics. Such grand partition functions are uniquely determined by an inessential multiplicative constant, and allow us to perform all standard thermodynamic computations.

{ A quite interesting fact is that the phenomenon of condensation of non negligible amount of particles in the fundamental state indeed occurs also and only for $q\in(0,1]$. 
This was demonstrated in a more general context involving the grand canonical ensemble, in~\cite{AF}, Section 5, by using  only \eqref{qcra} and imposing the KMS condition, in the mathematical setting of the distributions. 

This allows to call $q$-particles, $q\in(0,1)$, as \emph{Bose-like particles} and \emph{Fermi-like particles} for $q\in(-1,0)$. The separation point $q=0$ corresponds to the Boltzmann (i.e.,\ classical) case, where it is well known that the condensation does not occur. 

Concerning the Fermi-like particles, we note that it seems to be totally meaningless to argue about an analogous ``Pauli exclusion principle'', simply because it is still unclear what is the statistics to which the quons obey.

}

In view of possible applications, we recover the equation of state for the perfect gas of $q$-particles, and discuss the occurrence of this condensation for Bose-like quons by recovering in a different way the formula \eqref{denscrit} that already appeared in \cite{AF}, Equation (5.6).

\section{The Grand Partition Function for Fermi, Bose and Boltzmann Models}
\label{1mic1}

We start from a system whose Hamiltonian $H$ is a positive selfadjoint operator with compact resolvent, acting on a separable Hilbert space $\ch$, called the \emph{one-particle space}.

In such a situation, the spectrum $\s(H)=\{\eps_i\}$ is made by isolated points, accumulating at $+\infty$ if $\ch$ is infinite dimensional. In addition, we denote by $g_i$ the (necessarily finite) \emph{multiplicity}, that is the \emph{degeneracy},  of each eigenvalue $\eps_i$. Summarising,
\begin{equation}
\label{rcha}
H=\sum_{\eps_i\in\s(H)}\eps_i P_{\eps_i},\,\,\text{and}\,\,g_i:=\text{dim}\left(\text{Ran}(P_{\eps_i})\right)<\infty\,.
\end{equation}

We also suppose that at any inverse temperature $\b:=1/k_B T$, 
$k_B\approx 1.3806488 \times 10^{-23}\, {\rm J K}^{-1}$ being the Boltzmann constant, $e^{-\b H}$ is trace-class and
define the {\it partition function}
$\z:=\ty{Tr}\!\left(e^{-\b H}\right)$. 

Concerning the \emph{grand partition function} $\ty{Z}$, it comes by considering open systems in thermodynamic equilibrium at inverse temperature $\b$ and {\it chemical potential} $\m$. It is  computed as standard for Fermi and Bose particles with the use of the symmetric and totally anti-symmetric (due to Pauli exclusion principle) Fock spaces $\gf_{\pm}(\ch)$, see \cite{BR}, Section 5.2.1. Indeed,
\begin{equation*}
\ty{Z}=\left\{
\begin{array}{lll}
e^{\dss -\tr\ln\left(I-e^{\b\m}e^{-\b H}\right)},& \m<\min\s(H) &\text{(Bose)}\,,\\
e^{\dss \tr\ln\left(I+e^{\b\m}e^{-\b H}\right)}, & \m\in\br &\text{(Fermi)}\,.
\end{array}
\right.
\end{equation*}
We indicate such grand partition functions as $\ty{Z}_{\pm 1}$, where $\pm 1$ correspond to the Bose/Fermi alternative. 
\begin{Prop}
\label{BFgp}
For $\ty{Z_{\pm 1}}$, we have the estimate
\begin{equation*}
\begin{array}{ll}
\ty{Z_{-1}}\leq e^{\dss \z e^{\b\m}}\,,&\m\in\br\,,\\
\ty{Z_1}\leq e^{\dss \left( \frac{\z e^{\b\m}}{ 1-e^{\b(\m-\min\s(H))}}\right)}\,,&\m<\min\s(H)\,.\\
\end{array}
\end{equation*}
\end{Prop}
\begin{proof}
Those are nothing other than the proofs of \cite{BR}, Proposition 5.2.22 (Fermi case) and Proposition 5.2.27 (Bose case) respectively.
\end{proof}
Under our assumptions, it is easy to recognise that the partition function is a smooth function of $\b$ and, once  all the eigenvalues of $H$ have been fixed but $\eps_{i_o}$, it is also a smooth function of each $\eps_{i_o}$.

By Proposition \ref{BFgp}, the grand partition functions are well defined whenever
the chemical potential $\m$ is arbitrary (Fermi case), or strictly less than the first eigenvalue $\min\s(H)$ of the Hamiltonian (Bose case).
As for the partition functions, the grand partition functions are smooth functions of the parameters $\b,\m$ and of the eigenvalues of $H$. Such smooth dependences on parameters allow us to compute many thermodynamic functions, see e.g.,~\cite{Hu}.
We often omit to indicate such dependences for the sake of simplicity.
\smallskip

The Boltzmann case is very particular because, in the Boltzmann statistics, the Gibbs paradox (e.g.,\ \cite{Hu}) naturally emerges, and thus we should suitably correct the statistical weights $w_{\eps_i}$, e.g.,\ \cite{FV} (p. 680).

In this situation, it might be natural to use the the so-called full Fock space $\gf(\ch)$ and the \emph{grand canonical Hamiltonian} $K:=\rmd\G(H)-\m N$, being $\rmd\G(H)$ the {\it second quantized} of the operator $H$ and 
$N$ the {\it number operator}, as for the computation of $\ty{Z}_{\pm 1}$ in the Bose and Fermi cases, see e.g.,\, \cite{BR}, Section 5.2.1. 
This computation is precisely what was done in \cite{W}, Formula (11), obtaining 
\begin{equation}
\label{gfip0} 
\tr\left(e^{-\b K}\right)=\frac{1}{1-\z e^{\b\m}}\,,
\end{equation}
holding true again for $\m<\min\s(H)$.

As we have already explained, such a formula is unrealistic for several reasons, the main one being
that the grand partition function for the Boltzmann statistics should be defined for any value of the chemical potential. 
However, as we will show below (see also~\cite{Hu}, Section 7.3), the correct formula should be
$\ty{Z_0}=e^{\z e^{\b\m}}$. 

We would like to note that, defining the \emph{fugacity}, or \emph{activity}, by $z:=e^{\b\m}>0$, we have~that
\begin{equation}
\label{gfip}
\tr\left(e^{-\b K}\right)=\sum_{n=0}^{+\infty}\z^n z^n\,,\quad  \z z<1\,.
\end{equation}

It is interestingly seen that, if one corrects \eqref{gfip} with the weight $n!$ in the denominator of all addenda of the series to avoid the Gibbs paradox, we obtain
\begin{equation}
\label{gfipp}
\sum_{n=0}^{+\infty}\frac{\z^n}{n!}z^n=e^{\z z}=\ty{Z_0}\,.
\end{equation}
 
It is also customary to express the grand partition function in terms of the  fugacity $z$,~obtaining
\begin{equation}
\label{gfipo}
\left\{
\begin{array}{lll}
\ty{Z_1}=e^{\dss-\tr\ln\left(I-ze^{-\b H}\right)}, &z<e^{\b\min\s(H)} & \text{(Bose)}\,, \\
\ty{Z_0}=e^{\dss z\tr\left(e^{-\b H}\right)},&z>0 & \text{(Boltzmann)}\,, \\
\ty{Z_{-1}}=e^{\dss\tr\ln\left(I+ze^{-\b H}\right)}, &z>0 & \text{(Fermi)}\,.
\end{array}
\right.
\end{equation}

Now, another incongruence immediately emerges concerning the computation of the grand partition function for the Boltzmann case. Indeed, suppose for simplicity that $\ch$ is finite dimensional, so all involved traces are finite sums. 

If one approximates the 1st and the 3rd lines of \eqref{gfipo} for $z\to0$, that is in the  low-density regime, we correctly obtain 
$\ty{Z_1}\approx\ty{Z_0}\approx\ty{Z_{-1}}$ because the Bose and Fermi distributions should coincide with the Boltzmann one for $z\approx0$. Therefore, the grand partition function for $q=0$ cannot have the form \eqref{gfip0}.

\section{The Grand Partition Function of Quons: One-Mode}

To discover what the grand partition function for quons should be, we first note that the computations in \cite{W} provide the wrong result for $q=0$. 
In addition, if one uses the grand canonical Hamiltonian $K=\rmd\G(H)-\m N$ acting on the $q$-deformed Fock space $\gf_q(\ch)$ (e.g.,\ \cite{BKS}), one obtains again \eqref{gfip0}, independently of $q\in(-1,1)$.
\smallskip

On the other hand, it was discovered, first in \cite{AF}, Section 5, for the grand canonical ensemble, and then in \cite{FV}, Section 2, using the microcanonical ensemble, that the generalisation to $q\in(-1,1)$ of the Planck distribution of the occupation numbers at inverse temperature $\b$ and fugacity $z$ is
\begin{equation}
\label{occq}
n_q(\eps)=g(\eps)\frac{1}{z^{-1}e^{\b\eps}-q}\,,\quad q\in[-1,1]\,,\,\, \eps\in\s(H)\,.
\end{equation}
Here, $g(\eps)$ is the degeneracy of the level $\eps$.
{
For the cases treated in Section \ref{continuu}, where the Hamiltonian is proportional to the opposite of the Laplacian $-\D$ acting on $L^2(\br^d)$, the degeneracy $g(\eps)$ is absorbed in an appropriate integral after passing to the continuum. 
}

The simplest way to try to compute such a grand partition function is to consider the so-called \emph{one-mode model} and, since it is connected with the so-called
\emph{$q$-oscillator} (e.g.,~\cite{EM}), we have $\ch=\bc$ and $H=\hbar\om$. 
It should be noticed that \eqref{occq} also appeared  in \cite{Good}, (A.5), apparently computed by using the \emph{q-numbers}
$[n]_q$  (see also e.g.,\ \cite{BKS}). Unfortunately, it is unclear how (A.5) in \cite{Good} is derived and, in addition, any computation of the 
$q$-grand partition function using the $[n]_q$ numbers 
does not reproduce \eqref{occq}.

Our approach in computing the grand partition function in one-mode model proceeds as follows, by taking into account that the degeneracy is obviously $1$. We start with the well-known formula in the one-mode, holding for values of the fugacity previously~determined,
\begin{equation}
\label{occq0}
n=z\frac{\partial\ln{\ty{Z}}_q}{\partial z}\,.
\end{equation}

Combining \eqref{occq} and \eqref{occq0}, we obtain
\begin{equation}
\label{occq1}
\frac{\partial\ln{\ty{Z}}_q}{\partial z}=\left\{
\begin{array}{lll}
e^{-\b\hbar\om},& z>0 & q=0\,, \\
\frac{\dss  1}{\dss e^{\b\hbar\om}-qz},& z>0  &-1<q<0\,,\\
\frac{\dss 1}{\dss e^{\b\hbar\om}-qz}, & 0<z<\frac{\dss e^{\b\hbar\om}}{\dss q} & 0<q<1\,.
\end{array}
\right.
\end{equation}

Integrating both members of \eqref{occq1}, and neglecting the inessential multiplicative constant, we obtain $\ty{Z_0}$ and $\ty{Z_{\pm1}}$  in \eqref{gfipo}.
The general case (i.e.,\ the multi-mode case) involving the computation of ${\ty{Z}}_q$ will be handled in the forthcoming section, obtaining \eqref{occq3}. 

\section{The Grand Partition Function of Quons}
\label{gpfqmm}

The present section is devoted to the general case of the multi-mode model  (i.e.,\ $\dim(\ch)>1$) 
described by the Hamiltonian in \eqref{rcha}, by following the previous suggestions.

Indeed, for $q\in(-1,1)\smallsetminus\{0\}$, define
\begin{equation}
\label{occq3}
{\ty{Z}}_q:=e^{\dss -\frac{\dss \tr\ln\left(I-zqe^{\dss -\b H}\right)}{\dss q}},\,\, 
\left\{\begin{array}{ll}
0<z<\frac{\dss e^{\dss \b\min\s(H)}}{\dss q}\,\,& 0<q<1\,, \\
{}&\\
z>0\,\,& -1<q<0\,.
\end{array}\right.
\end{equation}

We note that:
\begin{equation*}
{\ty{Z}}_q(z):=
\left\{\begin{array}{ll}
\left(\ty{Z_1}(zq)\right)^\frac{1}{q},\,\, &0<q<1\,, \\
\left(\ty{Z_{-1}}(z|q|)\right)^\frac{1}{|q|},\,\, &-1<q<0\,,
\end{array}\right.
\end{equation*}
and thus all such grand partition functions are well defined in their own domain involving the activity $z$, and we have the following estimate for the
${\ty{Z}}_q$.
\begin{Thm} 
\label{Qgp}
For ${\ty{Z}}_q$, $q\in[-1,1]$, we have the estimate
\begin{equation*}
\begin{array}{lll}
{\ty{Z}}_q\leq e^{\dss \z z}\,,& 0<z& q\in[-1,0]\,,\\
{\ty{Z}}_q\leq e^{\dss \left(\frac{\z zq}{ 1-zqe^{-\b\min\s(H)}}\right)}\,,& 0<z<\frac{\dss e^{\b\min\s(H)}}{\dss q}&  q\in(0,1]\,.
\end{array}
\end{equation*}
\end{Thm}
\begin{proof}
For $q=0$, the first inequality saturates to the Boltzmann case. For the negative cases $q\in[-1,0)$, we use the 2nd line in \eqref{occq3} and
Proposition \ref{BFgp} (or take advantage of the easy inequality
$\ln(1+x)\leq x$), obtaining the 1st row. For the case $q\in(0,1]$, we use the 1st line in \eqref{occq3}, obtaining the 2nd row with the bound for the activity $z$ by taking into account Proposition \ref{BFgp}.
\end{proof}

The definition of ${\ty{Z}}_q$ in \eqref{occq3} is justified by the following
\begin{Prop}
\label{PropN}
For each $\eps\in\s(H)$ fixed,  ${\ty{Z}}_q$ with $q\in[-1,1]$, is derivable w.r.t. $\eps$ and we have
$$
-\frac1{\b}\frac{\partial\ln{\ty{Z}}_q}{\partial\eps}=n_q(\eps)\,,
$$
where the $n_q(\eps)$ are the occupation numbers in \eqref{occq}. 
\end{Prop}
\begin{proof}
For $q=0$, it is enough to show that the partition function $\z$ admits partial derivatives w.r.t.\ any $\eps\in\s(H)$ in a suitable neighbourhood. 

Indeed, for $\eps_o\in\s(H)$, 
and for $x$ in the non empty open interval $I_{\eps_o}$ centered in $\eps_o$, such that $I_{\eps_o}\bigcap\big(\s(H)\smallsetminus\{\eps_o\}\big)=\emptyset$, which always exists as $\s(H)$ is discrete, define
$$
H_\eps(x):=xP_{\eps_o}+H_{\eps_o}=xP_{\eps_o}+\sum_{\eps\neq\eps_o}\eps P_{\eps}\,.
$$

It is clear that $\z=\z(x)$ as a function of $x$ around $\eps_o$, is equal to
$$
\z(x)=g(\eps_o)e^{-\b x}+\tr\big(e^{-\b H_{\eps_o}}\big)\,,
$$
where the 2nd piece in the r.h.s.\ does not depend on $x$.

Therefore, $\z(x)$ is smooth in the neighbourhood $I_{\eps_o}$ and
$$
\frac{\partial\z}{\partial\eps_o}=\frac{\di\z(x)}{\di x}\Big|_{x=\eps_o}=-\b g(\eps_o)\e{-\b\eps_o}\,.
$$
We then get $-\frac1{\b}\frac{\partial\ln\ty{Z_0}}{\partial\eps}=n_0(\eps)$. 

The computation for the remaining $q$ is similar. Indeed, suppose $q\neq0$, then
\begin{align*}
\ln{\ty{Z}}_q(x)&=-\frac{\tr\big(\ln(I-zqe^{-\b H})\big)}{q}\\
&=-\frac{g(\eps_o)}{q}\ln(1-zqe^{-\b x})-\frac1{q}\tr\big(\ln(I-zqe^{-\b H_{\eps_o}})\big)
\end{align*}
where, as before, the 2nd piece in the r.h.s.\ does not depend on $x$. Then
\[
-\frac1{\b}\frac{\partial \ln{\ty{Z}}_q}{\partial\eps_o}=-\frac1{\b}\frac{\di\ln{\ty{Z}}_q(x)}{\di x}\Big|_{x=\eps_o}=n_q(\eps_o)\,.
\]
\end{proof}

To conclude with the grand partition function of the $q$-particles, we get the following result about the convergence of ${\ty{Z}}_q\equiv {\ty{Z}}_q(\b,z)$ to $\ty{Z_0}(\b,z)$.

\begin{Prop}
For the $q$-grand partition function ${\ty{Z}}_q$, we get
$$
\lim_{q\to 0} {\ty{Z}}_q (\b,z)=\ty{Z_0}(\b,z)\,,
$$
where the convergence is uniform in the variables $\b, z$ in all closed strips $[\b_o,+\infty)\times[0,\d]$,
where $\b_o>0$ and $\d<e^{\b_o\min\s(H)}$.
\end{Prop}
\begin{proof}
We start by noticing that, under the limitations on the variables $\b$ and $z$, the grand partition functions are defined for all values of $q\in[-1,1]$.
It will be enough to manage the logarithm of the ${\ty{Z}}_q(\b,z)$. For such a purpose, we use the 2nd order Mc Laurin expansion,~obtaining
\begin{align*}
|\ln{\ty{Z}}_q(\b,z)&-\ln\ty{Z_0}(\b,z)|=\left|\frac{\tr\ln(I-zqe^{-\b H})}{q}+z\tr e^{-\b H}\right|\\
&\leq\frac{|q|}2\big(z\tr e^{-\b H}\big)^2\leq\frac{|q|}2\big(\d\tr e^{-\b_o H}\big)^2\,.
\end{align*}
\end{proof}
In view of applications to thermodynamics, from the grand partition function we can recover the {\it Landau potential}, called also {\it grand potential}
\begin{equation*}
\label{granpot}
\Omega_q:=-\frac{1}{\b}\ln {\ty{Z}}_q\,,\quad q\in[-1,1]\,.
\end{equation*} 

It is well known that, from a thermodynamical point of view, $\Omega_q$ can be expressed by $\Omega_q=-PV$, where $V$ is the volume occupied by a real physical system under consideration.
Therefore, we obtain
\begin{equation}
\label{pivu0}
PV=-k_B T\frac{\tr\big(\ln(I-zqe^{-\b H})\big)}{q}\,,\quad q\in[-1,1]\smallsetminus\{0\}\,,
\end{equation} 
which reduces to $PV=k_B Tz\,\tr\big(e^{-\b H}\big)$ for $q=0$. 
\smallskip

On the other hand, it is well-known that the average number $N$ of particles of such systems is
$N=-z\frac{\partial\Omega_q}{\partial z}$, which immediately leads to the well-known equation of state $PV=Nk_{\rm B}T$ for a gas of free classical particles. 

For the remaining cases $q\in[-1,1]\smallsetminus\{0\}$, it is no longer true that an equation of state can be immediately obtained. In fact, also by considering directly \eqref{occq}, we have
\begin{equation}
\label{pivu}
\begin{split}
N&=-z\frac{\partial\Omega_q}{\partial z}=\sum_{\eps\in\s(H)}n_q(\eps) \\ 
&=\sum_{\eps\in\s(H)}g(\eps)\frac{1}{z^{-1}e^{\b\eps}-q}
=\tr\Big(\big(z^{-1}e^{\b H}-q\big)^{-1}\Big)\,,
\end{split}
\end{equation}
and thus it is unclear how to compare \eqref{pivu0} with \eqref{pivu} to obtain analogous equations of state for quons, bosons and fermions included.

In the next section, we make  such a discussion for the free gas of quons living in $\br^3$, after passing to the continuum.

\section{The Free Gas of Quons}
\label{continuu}

As a simple application of the grand partition function computed in Section \ref{gpfqmm}, we determine the equation of state for the free gas of quons. In order to do that, we  first have to perform the continuum limit to handle Hamiltonians with a continuous spectrum. 

Indeed, we consider the one-particle Hamiltonian $H:=-\frac{\D}{2m}$, acting on $L^2(\br^3,\di^3x)$. It is a selfadjoint operator with
$\s(H)=\s_{\rm ac}(H)=[0,+\infty)$, where $\s_{\rm ac}$ is \emph{ the absolutely continuous spectrum}. 
In this way, the degeneracy factor $g(\eps)$, $\eps\in\s(H)$, appearing in the computation of the canonical trace of 
$\cb(\ch)$ in \eqref{pivu0} and \eqref{pivu}, is absorbed in the integration on the whole spectrum.

Summarising, the limit to the continuum is now obtained in the following standard way, as explained, e.g.,\ in the appendix of \cite{FV}:
for a gas with $N$ particles, we simply make the replacements
\begin{equation}
\label{b4}
\sum_{\eps\in\s(H)}g(\eps)=\sum_{\{n(\eps)\mid\eps\in\s(H)\}}\rightarrow
V\int\frac{\di^3{\bf p}}{h^3}\,,
\end{equation}
$h\approx 6.626070040\times 10^{-34}$ Js being the Planck constant, and $V$ the volume of the physical system under consideration.

Hence, we start with the relation \eqref{pivu0} for the grand potential ${\ty{Z}}_q$,
and to relay our result to the standard equations of state for the Fermi or Bose particles, we distinguish between the Fermi-like and the Bose-like cases. The limit to the continuum is simply achieved by applying \eqref{b4}.
\smallskip

For the Fermi-like case, i.e.,\ $-1\leq q<0$, we obtain 
\begin{equation}
\label{PcontFERMI}
\begin{split}
\frac{\dss P}{k_B T}= 
\frac{1}{|q|}\frac{4\pi}{h^3}\int_0^\infty dp\, p^2 \ln\left(1+ z|q| e^{- \beta\frac{p^2}{2m}}\right)
= \frac{1}{|q|}\frac{1}{\lambda^3} f_{\frac{5}{2}}(z|q|),
\end{split}
\end{equation}
where $\lambda=\sqrt{2\pi h^2/mk_BT}$ is the \emph{thermal wavelength},
and the function $f_{\frac{5}{2}}$ is a well-known generalisation of the $\zeta$-function, similarly to $f_{\frac{3}{2}}$,  $g_{\frac{5}{2}}$ and $g_{\frac{3}{2}}$ below.

The equation of state may be obtained considering the \emph{relative volume per particle}, i.e.,\ $v:=V/N$, with $N$ in \eqref{pivu}, passing to the continuum limit as
\begin{equation}
\label{1/vcontFERMI}
\frac{\dss 1}{v}=\frac{\dss N}{V}= 
\frac{1}{|q|}\frac{4\pi}{h^3}\int_0^\infty dp\, p^2 \frac{\dss 1}{(z|q|)^{-1}e^{\frac{\beta p^2}{2m}}+1}
= \frac{1}{|q|}\frac{1}{\lambda^3} f_{\frac{3}{2}}(z|q|)\,.
\end{equation}

Collecting together \eqref{PcontFERMI} and \eqref{1/vcontFERMI}, the equation of state of a free gas of Fermi-like quons assumes the form 
\begin{equation}
\label{eosFERMI}
\frac{\dss PV}{k_B T}=N \frac{f_{\frac{5}{2}}(z|q|)}{f_{\frac{3}{2}}(z|q|)}\,.
\end{equation}

Notice that, for $z\approx 0$, we get back to the equation of state of the classical (i.e.,\ Boltzmann) perfect gas $\frac{PV}{k_B T}=N$.

For the Bose-like case, that is $0<q\leq1$, we first note that the condensation of particles into the ground state can occur as well, see \cite{AF, FV}. 
Before passing to the continuum, it is customary to separate the part corresponding to $\eps=0$ from the remaining one, obtaining
\begin{align}
\label{PcontBOSE}
\frac{\dss P}{k_B T}&= 
-\frac{1}{q}\left(\frac{4\pi}{h^3}\int_0^\infty dp\, p^2 \ln\left(1- zq e^{- \frac{\beta p^2}{2m}}\right)
+\frac{a}{V}\ln(1-zq)\right)\\ \nn
&=\frac{1}{q}\left(\frac{1}{\lambda^3} g_{\frac{5}{2}}(zq)-\frac{a}{V}\ln(1-zq)\right)\,,
\end{align}
\begin{align}
\label{1/vcontBOSE}
\frac{\dss 1}{v}&=
\frac{1}{q}\left(\frac{4\pi}{h^3}\int_0^\infty dp\, p^2 \frac{\dss 1}{(zq)^{-1}e^{\frac{\beta p^2}{2m}}-1} + \frac{a}{V} \frac{zq}{1-zq}\right)\\ \nn
&= \frac{1}{q}\left(\frac{1}{\lambda^3} g_\frac{3}{2}(zq)+ \frac{a}{V} \frac{zq}{1-zq}\right)\,.
\end{align}

In the condensation regime, we deal with a multi-phase situation, and thus the portion of the condensate, possibly also vanishing, can vary according to many a priori constrains, such as the boundary conditions used to reach the thermodynamic limit, see e.g.,\ \cite{BR}, \mbox{Section 5.2.5}, and \cite{F4}.
Therefore, in \eqref{PcontBOSE} and consequently in \eqref{1/vcontBOSE}, we introduced the dimensionless multiplicative constant 
$a\geq0$.

It is worth noticing that, when the portion of the condensate in the system vanishes, that is $a=0$, we can obtain the equation of state similar to the 
Fermi-like case \eqref{eosFERMI}. Indeed, collecting together  \eqref{PcontBOSE} and \eqref{1/vcontBOSE} with $a=0$, we obtain
$$
\frac{\dss PV}{k_B T}=N \frac{g_{\frac{5}{2}}(zq)}{g_{\frac{3}{2}}(zq)}\,,
$$
which again reduces to the equation of state for the Boltzmann free gas for small fugacity.

We end by computing the critical density for Bose-like quons as in \cite{Hu}, Equation~(12.51), obtaining in \eqref{1/vcontBOSE} again with $a=0$,
\begin{equation}
\label{denscrit}
\r_{\rm c}(q)=\frac{1}{v_{\rm c}(q)}=\lim_{z\uparrow(1/q)}\frac{g_\frac{3}{2}(zq)}{q\l^3}=\frac{1}{q}\frac{g_\frac{3}{2}(1)}{\l^3}=\frac{\r_{\rm c}(1)}{q}\,,
\end{equation}
which coincides with (5.6) in \cite{AF}.

\section{Conclusions}

It is well known that the main ingredient to deal with the thermodynamics of macroscopic systems in terms of statistical mechanics is the grand partition function, which in the Fermi/Bose cases can be computed by the standard techniques of second quantisation. 

As explained through the present paper, such a standard method fails, even in the case of the free gas of classical particles, 
that is obeying to the Boltzmann statistics. However, in the Boltzmann case it will be enough to take into account the Gibbs correction factor $1/n!$ to count in the right way the computations involving the level of $n$-particles in the full Fock~space. 

For the remaining cases $q\in(-1,0)\bigcup(0,1)$, there is no reasonable indication to establish a similar ansatz. In other words, it is totally unknown how  the reasonable statistics for the $q$-particles should be
(apart some attempts toward this direction have been carried out by \cite{Fiv,St}), 
which produce a reasonable thermodynamical behaviour for such exotic systems, supposing that these really exist in nature.

On the reverse side, just supposing that the quons might play some role for some application in quantum physics, and also for example in quantum information and quantum computing, several attempts have been made to produce a decent thermodynamics for such particles, sometimes also reaching  paradoxical conclusions.

In the present paper, we have computed the, unknown up to now, grand partition function for all such $q$-particles, by finding out those of Fermi and Bose particles, and the correct function for the Boltzmann situation, as particular cases. 

As a simple application, we discussed the corresponding equations of the state, showing that all those equations reduce themselves to that of the classical gas $PV=Nk_B T$ for small fugacity $z\approx 0$. We also briefly discuss the condensation phenomena for Bose-like quons (i.e.,\ $q\in(0,1]$) whose appearance is rigorously proved in \cite{AF}, by finding the formula for the critical density $\r_c(q)$ already obtained in the just mentioned paper with different~methods.

It is then evident that the thermodynamics of quons can now be carried out, once having computed the corresponding grand partition function, see e.g.,\ \cite{Hu, LL}. In particular, it would be desirable to investigate the thermodynamical properties enjoyed by Fermi-like quons (i.e.,\ $q\in[-1,0)$), which are completely different from those of Bose-like quons. This detailed analysis is out of the scope of the present paper and is left out for the interested~readers.

However, as shown in \cite{AF,FV} in a rigorous way, the Bose-like quons share with bosons the phenomenon of the condensation into the fundamental state. In view of theoretical and concrete application of this interesting fact, it might be of interest to extend the investigation of quons living in inhomogeneous networks following the lines \mbox{of \cite{F1, F3, F4}}. It is also of interest to investigate the magnetic properties of systems of quons  on lattices analogously to the celebrated Ising and similar models, including the disordered ones, see e.g.,\ \cite{BF0, BF, FM} and the literature cited therein. We leave out these interesting arguments, postponing these for future investigation.

In view of possible applications, another direction is to start to investigate the thermodynamics of other systems enjoying exotic commutation relations. Among them, we cite the Boolean and monotone ones, see e.g.,\ \cite{CFL} and the literature cited therein. 

Concerning these last two cases, the unique available method to compute the grand partition function is the second quantisation one, involving the relative Fock space construction. We briefly discuss the simpler Boolean case, leaving the monotone one for a forthcoming work. 

Indeed, for the Boolean case, the creators and annihilators satisfy the \mbox{commutation~relations}
\begin{equation}
\label{boole}
{\bf b}(f){\bf b}^\dagger(g)+\langle g,f\rangle \sum_{j\in J}{\bf b}^\dagger(e_j){\bf b}(e_j)=\langle g,f\rangle I_\ch\,, \quad f,g\in\ch\,.
\end{equation}
Here, $\{e_j\mid j\in J\}\subset\ch$, ${\rm card}(J)$ coinciding with the Hilbertian dimension $\dim(\ch)$ of $\ch$, is any orthonormal basis of $\ch$, and the possibly infinite sum in \eqref{boole} is meant as in \cite{CFL}, Proposition 3.2.

The associate Boolean Fock space is given as $\gf_{\rm boole}(\ch)=\bc\,\Om\oplus\ch$, $\Om$ being the unit vacuum vector, and thus we immediately compute the associated grand partition function $Z_{\rm boole}=1+z\tr\,e^{-\b H}$.

Concerning the average occupation numbers, we get 
$$
n(\eps)=g(\eps)\frac{ze^{-\b\eps}}{1+z\tr\,e^{-\b H}}\leq N=\frac{z\tr\,e^{-\b H}}{1+z\tr\,e^{-\b H}}<1\,,
$$
where the 1st inequality is always sharp, but in the trivial case $\dim(\ch)=1$.

The explanation of the physical interpretation of this result can be found in \cite{VW}, in which the Boolean statistics was first introduced in relation to the application to quantum~optics.
\vspace{6pt}

%
%
%
%

\noindent
\textbf{Acknowledgments}.
The authors acknowledge MIUR Excellence Department Project awarded to the
Department of Mathematics, University of Rome ``Tor Vergata'', CUP
E83C18000100006, and Italian INdAM-GNAMPA. 
The first-named author is partially supported by MIUR-FARE R16X5RB55W QUEST-NET and by the University of Rome
``Tor Vergata'', funding scheme ``Beyond Borders'', CUP E84I19002200005.






\begin{thebibliography}{999}

\bibitem{Ha}
Haag, R.
\textit{Local Quantum Physics: Fields, Particles, Algebras}, 2nd ed.;
Springer: Berlin, Germany, 1996; p. 405.

\bibitem{Sc}
Scharf, G. Physica A. In
\textit{Finite	Quantum	Electrodynamics: The	Causal	Approach}, 3rd ed.;
Dover Publication, Inc.: Mineola, NY, USA, 2014; p. 392.

\bibitem{CRV}
Ciolli, F.; Ruzzi, G.; Vasselli, E.
QED representation for the net of causal loops.
{\em Rev. Math. Phys.} {\bf 2015}, {\em 27}, 104416, 1--35. [\href{http://doi.org/10.1142/S0129055X15500129}{CrossRef}]

\bibitem{BCRV}
Buchholz, D.; Ciolli, F.; Ruzzi, G.; Vasselli, E.
The universal $C^*$-algebra of the electromagnetic field II.
Topological charges and spacelike linear fields.
{\em Lett. Math. Phys.}  {\bf 2017}, {\em 107}, 201--222. { [\href{http://dx.doi.org/10.1007/s11005-016-0931-x}{CrossRef}]
\bibitem{BCRV19}
Buchholz, D.; Ciolli, F.; Ruzzi, G.: Vasselli, E.
On string-localized potentials and gauge fields.
{\em Lett. Math. Phys.}  {\bf 2019}, {\em 109}, 829--842.
} [\href{http://dx.doi.org/10.1007/s11005-018-1136-2}{CrossRef}]


\bibitem{CL}
Connes, A.; Lott, J.
Particle models and noncommutative geometry.
{\em  Nuclear Phys. B Proc. Suppl.} {\bf 1991}, {\em 18B}, 29--47. [\href{http://dx.doi.org/10.1016/0920-5632(91)90120-4}{CrossRef}]

\bibitem{A}
Araki, H.
Relative entropy of states of von Neumann algebras.
{\em Publ. RIMS Kyoto Univ.} {\bf 1976},  {\em 11}, 809--833. [\href{http://dx.doi.org/10.2977/prims/1195191148}{CrossRef}]

\bibitem{CLR}
Ciolli, F.; Longo, R.; Ruzzi, G.
The Information in a Wave.
{\em Commun. Math. Phys.}  {\bf 2020}, {\em 379}, 979--1000. [\href{http://dx.doi.org/10.1007/s00220-019-03593-3}{CrossRef}]

\bibitem{CLRR}
Ciolli, F.; Longo, R.; Ranallo, A.; Ruzzi, G.
Relative entropy and curved spacetimes.
{\em J. Geometry Phys.} {\bf 2022}, {\em 172}, 104416. [\href{http://dx.doi.org/10.1016/j.geomphys.2021.104416}{CrossRef}]

\bibitem{T}
Tsallis, C.
Possible generalization of Boltzmann-Gibbs statistics.
{\em J. Stat. Phys.} {\bf 1988}, {\em 52}, 479--487. [\href{http://dx.doi.org/10.1007/BF01016429}{CrossRef}]

\bibitem{T1}
Tsallis, C.
Thermodynamics and statistical mechanics for complex systems---Foundations
and applications.
{\em Acta Phys. Polonica B} {\bf 2015}, {\em 46}, 1089. [\href{http://dx.doi.org/10.5506/APhysPolB.46.1089}{CrossRef}]


\bibitem{FMa}
Fr\"ohlich, J.;  Marchetti, P.A.
Quantum field theories of vortices and anyons.
{\em Commun. Math. Phys.} {\bf 1989}, {\em 121}, 177--223. [\href{http://dx.doi.org/10.1007/BF01217803}{CrossRef}]


\bibitem{BLW}
Bo\.zejko, M.; Lytvynov, E.; Wysocza\'nski J.
Noncommutative L\'evy processes for generalized (particularly anyon) statistics.
{\em Commun. Math. Phys.} {\bf 2012}, {\em 313}, 535--569. { [\href{http://dx.doi.org/10.1007/s00220-012-1437-8}{CrossRef}]
\bibitem{BL}
Biedenharn, L.C.; Lohe, M.A.
\textit{Quantum Group Symmetry and Q-Tensor Algebras};
World Scientific: Singapore, 1995; p. 304.
}

{
\bibitem{BS}
Bo\.zejko, M.; Speicher R.
Completely positive maps on Coxeter groups, deformed commutation relations, and operator spaces.
{\em Math. Ann.} {\bf 1994}, {\em 300}, 97--120.
} [\href{http://dx.doi.org/10.1007/BF01450478}{CrossRef}]

\bibitem{S}
Shlyakhtenko, D.
Free Quasi-Free States.
{\em Pacific J. Math.} {\bf 1997}, {\em 177}, 329--368. [\href{http://dx.doi.org/10.2140/pjm.1997.177.329}{CrossRef}]

\bibitem{W}
Werner, R.F.
The free quon gas suffers Gibbs' paradox.
{\em Phys. Rev. D} {\bf 1993}, {\em 48}, 2929, 1--6. [\href{http://dx.doi.org/10.1103/PhysRevD.48.2929}{CrossRef}]

\bibitem{Gr}
Greenberg, O.W.
Particles with small violations of Fermi or Bose statistics.
{\em Phys. Rev. D} {\bf 1991}, {\em 43}, 4111--4120. [\href{http://dx.doi.org/10.1103/PhysRevD.43.4111}{CrossRef}] [\href{http://www.ncbi.nlm.nih.gov/pubmed/10013376}{PubMed}]

\bibitem{AK}
Avancini, S.S.; Krein, G.
Many-body problems with composite particles and $q$-Heisenberg algebras.
{\em J. Phys. A Math. Gen.} {\bf 1995},  {\em 28}, 685--691. [\href{http://dx.doi.org/10.1088/0305-4470/28/3/021}{CrossRef}]

\bibitem{M}
M\o ller, J.S.
Second quantization in a quon-algebra.
{\em  J. Phys. A Math. Gen.} {\bf 1993}, {\em 26}, 4643--4652. [\href{http://dx.doi.org/10.1088/0305-4470/26/18/028}{CrossRef}]

\bibitem{IK}
Inomata, A.;  Kirchner, S.
Bose-Einstein condensation of a quon gas.
{\em Phys. Lett. A} {\bf 1997}, {\em 231}, 311--314. [\href{http://dx.doi.org/10.1016/S0375-9601(97)00345-9}{CrossRef}]


\bibitem{VDN}
Voiculescu, D.V.; Dykema, K.J.; Nica, A.
\textit{Free Random Variables};
American Mathematical Soc.: New York, NY, USA, 1992; p. 70.

\bibitem{BKS}
Bo\.zejko, M.; K\"ummerer, B.; Speicher, R.
q-Gaussian Processes: Non-commutative and Classical Aspects.
{\em Commun. Math. Phys.} {\bf 1997}, {\em 185}, 129--154. [\href{http://dx.doi.org/10.1007/s002200050084}{CrossRef}]

\bibitem{LL}
Landau, L.D.;  Lifshits, E.M.
\textit{Statistical Physics}, 3rd ed.;
Course of Theoretical Physics; Butterworth-Heinemann: Oxford, UK, 2008; Volume 5, p. 564.

\bibitem{Hu}
Huang, K. \textit{Statistical Mechanics}, 2nd ed.;
John Wiley \& Sons: New York, NY, USA, 1987; p. 506.

\bibitem{FV}
Fidaleo, F.; Viaggiu, S.
A proposal for the thermodynamics of certain open systems.
{\em  Phys. A Stat. Mech. Its Appl.}
{\bf 2017}, {\em 468}, 677--690. [\href{http://dx.doi.org/10.1016/j.physa.2016.10.058}{CrossRef}]

\bibitem{AF}
Accardi, L.; Fidaleo, F.
Bose-Einstein condensation and condensation of q-particles in equilibrium and nonequilibrium thermodynamics.
{\em Rep. Math. Phys.} {\bf 2016},  {\em 77}, 153--182. [\href{http://dx.doi.org/10.1016/S0034-4877(16)30018-0}{CrossRef}]


\bibitem{BR}
Bratteli, O.; Robinson, D.W.
\textit{Operator Algebras and Quantum Statistical Mechanics 2};
Springer: Berlin/Heidelberg, Germany, 1997; p. 517.

\bibitem{EM}
Emerin, V.V.;  Meldianov, A.A.;
The $q$-deformed harmonic oscillator, coherent states, and the uncertainly relation.
{\em Theor. Math. Phys.} {\bf 2006}, {\em 147}, 709--715.

\bibitem{Good}
Goodison, J.W.;  Toms D.J.
The canonical partition function for quons.
{\em Phys. Lett. A} {\bf 1994}, {\em 195}, 38--42. [\href{http://dx.doi.org/10.1016/0375-9601(94)90423-5}{CrossRef}]

\bibitem{F4}
Fidaleo, F.
Harmonic analysis on inhomogeneous amenable networks and the Bose–Einstein condensation.
{\em J. Stat. Phys.} {\bf 2015}, {\em 160}, 715--759. [\href{http://dx.doi.org/10.1007/s10955-015-1263-4}{CrossRef}]

\bibitem{Fiv}
Fivel, D.I. Interpolation between Fermi and Bose statistics using generalized commutators.
{\em Phys. Rev. Lett.} {\bf 1990}, {\em 65}, 3361--3364;
Erratum in {\bf 1992}, {\em 69}, 2020. [\href{http://dx.doi.org/10.1103/PhysRevLett.65.3361}{CrossRef}]

\bibitem{St}
Stanley, R.P.
\textit{Enumerative Combinatorics}, 2nd ed.;
Cambridge University Press: Cambridge, UK, 2012; Volume 1,  p. 642.

\bibitem{F1}
Fidaleo, F.
Harmonic analysis on perturbed Cayley trees.
{\em J. Funct. Anal.} {\bf 2011}, {\em 261}, 604--634. [\href{http://dx.doi.org/10.1016/j.jfa.2011.04.007}{CrossRef}]

\bibitem{F3}
Fidaleo, F.
Harmonic analysis on Cayley trees II: The Bose–Einstein condensation.
{\em Infin. Dimens. Anal. Quantum Probab. Relat. Top.}
{\bf 2012}, {\em 15},  1250024. [\href{http://dx.doi.org/10.1142/S0219025712500245}{CrossRef}]

\bibitem{BF0}
Barreto, S.D.;  Fidaleo, F.
On the structure of KMS states of disordered systems.
{\em Commun. Math. Phys.}  {\bf 2004}, {\em 250}, 1--21. [\href{http://dx.doi.org/10.1007/s00220-004-1137-0}{CrossRef}]

\bibitem{BF}
Barreto, S.D.;  Fidaleo, F.
Disordered Fermions on Lattices and Their Spectral Properties.
{\em J. Stat. Phys.}
{\bf 2011}, {\em 143}, 65--684. [\href{http://dx.doi.org/10.1007/s10955-011-0197-8}{CrossRef}]

\bibitem{FM}
Fidaleo, F.; Mukhamedov, F.
Diagonalizability of non homogeneous quantum Markov states and
associated von Neumann algebras.
{\em Probab. Math. Stat.}
{\bf 2004}, {\em 24}, 401--418.


\bibitem{CFL} Crismale, V.; Fidaleo, F.; Lu, Y.G.
Ergodic theorems in quantum probability: An application to the monotone stochastic processes.
{\em Ann.  Sc.  Norm.  Sup. Pisa Cl.  Sci.}
{\bf 2017}, {\em 17}, 113--141. [\href{http://dx.doi.org/10.2422/2036-2145.201506_009}{CrossRef}]

\bibitem{VW}
von Waldenfels, W.
An approach to the theory of pressure broadening of spectral lines.
In {\em Probability and Information Theory II}; Lect. Notes Math.; Behara, M., Krickeberg,  K., Wolfowitz, J.,  Eds.;
Springer: Berlin/Heidelberg, Germany; New York, NY, USA, 1973; Volume 296, pp. 19--69.

\end{thebibliography}


%

\vspace{1,5cm}

\end{document}